\documentclass[aps,prl,twocolumn,superscriptaddress,floatfix,letter,longbibliography]{revtex4-2}

\usepackage{float}

\usepackage[T1]{fontenc}
\usepackage[utf8]{inputenc}
\usepackage{url}
\usepackage{amsmath}
\usepackage{amsthm}
\usepackage{amssymb}
\usepackage{mathtools}
\usepackage{bbold}
\usepackage{mathrsfs}

\usepackage{graphicx}
\usepackage{epstopdf} 
\usepackage{wrapfig}
\usepackage{braket}
\usepackage{xcolor}
\usepackage{comment}
\usepackage{lipsum}
\usepackage{hyperref}
\usepackage{dsfont}

\newcommand\eref[1]{Eq.~\ref{#1}}
\newcommand\aref[1]{Appendix~\ref{#1}}
\newcommand\fref[1]{Fig.~\ref{#1}}

\newcommand{\sx}{\sigma^{\rm x}}

\newcommand{\sz}{\sigma^{\rm z}}
\global\long\def\ii{\mathrm{i}}%
\newcommand{\im}{\mathrm{i}}

\newtheorem{theorem}{Theorem}

\begin{document}

\title{Variational ground-state quantum adiabatic theorem}

\author{Bojan \v{Z}unkovi\v{c}}
\email{bojan.zunkovic@fri.uni-lj.si}
\affiliation{
University of Ljubljana,  
Faculty of Computer and Information Science, Ve\v{c}na pot 113, 1000 Ljubljana, Slovenia
}
\author{Pietro Torta}
\affiliation{Dipartimento di Fisica, Università degli Studi di Milano, Via Celoria 16, 20133 Milano, Italy}
\affiliation{SISSA, Via Bonomea 265, I-34136 Trieste, Italy}
\author{Giovanni Pecci}
\affiliation{CNR-IOM, Via Bonomea 265, I-34136 Trieste, Italy}
\author{Guglielmo Lami}
\affiliation{Laboratoire de Physique Th\'eorique et Mod\'elisation, CNRS UMR 8089,
CY Cergy Paris Universit\'e, 95302 Cergy-Pontoise Cedex, France}
\author{Mario Collura}
\affiliation{SISSA, Via Bonomea 265, I-34136 Trieste, Italy}
\affiliation{INFN Sezione di Trieste, via Bonomea 265, I-34136 Trieste, Italy}

\begin{abstract} 
We present a variational quantum adiabatic theorem, which states that, under certain assumptions, the adiabatic dynamics projected onto a variational manifold follow the instantaneous variational ground state. We focus on low-entanglement variational manifolds and target Hamiltonians with classical ground states.
Despite the presence of highly entangled intermediate states along the exact quantum annealing path, the variational evolution converges to the target ground state. We demonstrate this approach with several examples that align with our theoretical analysis.
\end{abstract}

\maketitle

\paragraph{Introduction.}
Adiabatic quantum computation (AQC) is a computational paradigm relying on a slow continuous change of an initial Hamiltonian $H_0$ to a final Hamiltonian $H_1$ over a total annealing time $T$. The adiabatic theorem~\cite{albash2018adiabatic} ensures that if the initial state is the ground state of $H_0$ and the evolution is sufficiently slow, the final state approximates the ground state of $H_1$. By strategically selecting $H_1$, we can ensure that its ground state serves as a solution to a specific computational problem.  
Generally, AQC is polynomially equivalent to gate-based quantum computation and thereby universal~\cite{aharonov2008adiabatic,mizel2007simple}.
A relevant subclass of AQC protocols using stoquastic Hamiltonians addresses classical optimization problems and has been proposed as a quantum version of thermal annealing, with quantum fluctuations replacing thermal fluctuations~\cite{apolloni1989quantum}.
This Quantum Annealing (QA) formulation is compatible with hardware platforms~\cite{Johnson_Nat11} and holds significant promise for technological applications~\cite{QA_industry}.

Classical simulation of QA relies on simulated quantum annealing based on the path-integral Monte Carlo~\cite{crosson2016simulated} sampling. However, the actual need for quantum resources and the classical simulability of quantum protocols~\cite{cerezo2023does} remains largely unanswered, with the potential to lead to effective quantum-inspired algorithms~\cite{crowley2014quantum, QFT-low-ent, han2018unsupervised, vzunkovivc2022deep, vzunkovivc2023positive}, often based on tensor network techniques~\cite{haegeman2014geometry,biamonte2017tensor,huggins2019towards}.

We take a step in this direction by showing that an effective evolution constrained in a low-entanglement manifold can faithfully replace a fully quantum adiabatic evolution if the initial and final (target) states are classical.
Our approach helps to elucidate the role of entanglement in quantum annealing.
Recent numerical works~\cite{lami2023quantum,sreedhar2022quantum} found that, under certain conditions, restricting the accessible Hilbert space to a variational manifold with low entanglement might even increase the probability of finding the target state, potentially acting as a regularization technique.
Indeed, while entanglement is necessary for general quantum computation, its role in quantum optimization is unclear. 
A similar observation has been made in a related optimization strategy, namely the quantum approximate optimization algorithm (QAOA)~\cite{chen2022much, PhysRevA.106.022423}. 

The rest of the paper presents the physical intuition behind the variation ground-state quantum adiabatic evolution and several examples.

\begin{figure}[!t]
    \centering
    \includegraphics[width=0.85\columnwidth]{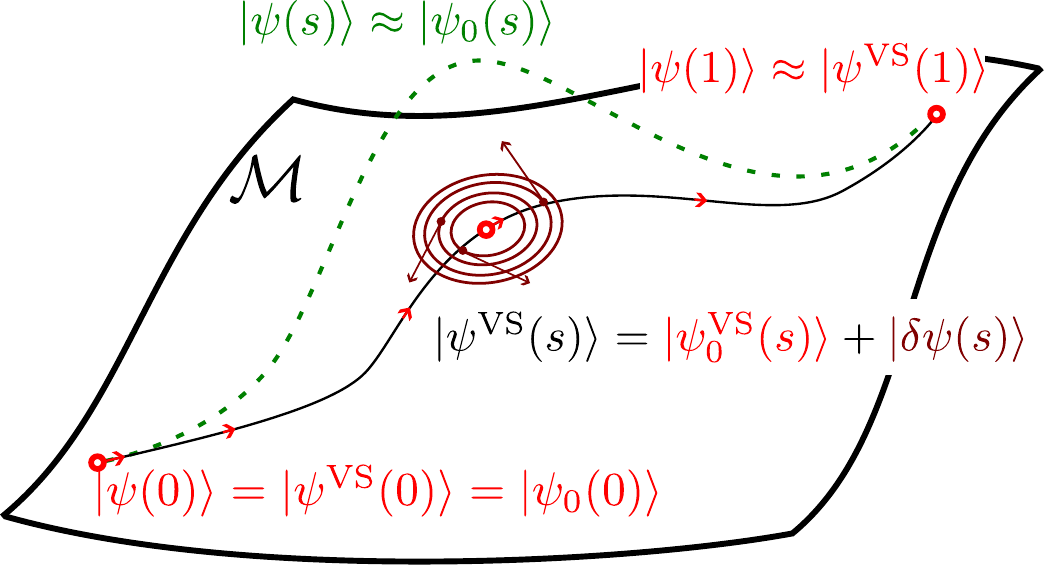}
    \caption{Schematic representation of the variational ground-state quantum adiabatic theorem. The dashed green line represents the exact quantum adiabatic evolution $\ket{\psi(s)} \approx \ket{\psi_0(s)}$, visiting high-entangled regions. Solid black line corresponds to the instantaneous variational ground state $\ket{\psi^{\rm VS}_0(s)}$, with the small red arrows indicating that it changes slowly with time. The time-dependent variational state $\ket{\psi^{\rm VS}(s)}$ is separated into $\ket{\psi^{\rm VS}_0(s)}$ and their difference $\ket{\delta\psi(s)}$. The latter can be decomposed in the eigenbasis of the linearized evolution map, whose eigenvalues $\omega_l$ determine the frequencies of the elliptic trajectories around the instantaneous variational ground state. Since $1/T \ll \omega_l$ (denoted by long, dark-red arrows), the effect of $\ket{\delta\psi(s)}$ averages to zero, and the time-dependent variational state $\ket{\psi^{\rm VS}(s)}$ follows the instantaneous variational ground state $\ket{\psi^{\rm VS}_0(s)}.$ 
    When the target ground state of the final Hamiltonian lies within the variational manifold, the time-dependent variational state will converge to it as the annealing time $T$ becomes large.}
    \label{fig: manifold volution}
\end{figure}

\paragraph{Variational ground-state quantum adiabatic theorem.}

Consider a parameterized Hamiltonian $H(s)$ interpolating between the initial Hamiltonian $H(0)=H_0$, and the final \textit{target} Hamiltonian $H(1)=H_1$.
Let us denote the duration of the protocol (annealing time) as $T$ and interpret $s=t/T$ as the rescaled time.
The evolution of a quantum state is then given by 
\begin{align}
    \frac{{\rm d}}{{\rm d}s}\ket{\psi(s)} = -{\rm i} T H(s)\,.
\end{align}
By initializing the system in the ground state of $H_0$, and if the annealing time $T$ is sufficiently large (i.e., under the quantum adiabatic condition~\cite{albash2018adiabatic}), the time-evolved state $\ket{\psi(s)}$ remains close to the ground state $\ket{\psi_0(s)}$ of the instantaneous Hamiltonian $H(s)$.
Since the exact description of $\ket{\psi(s)}$ is often challenging, we may approximate its real-time evolution with a variational state $\ket{\psi^{\rm VS}(s)}$ constrained to a manifold $\mathcal{M}$, evolved with the Lagrangian action variational principle~\cite{hackl2020geometry}.
Let us assume that the initial state, namely $\ket{\psi(0)}=\ket{\psi_0(0)}$,
lies within the variational manifold $\mathcal{M}$. Thus, the exact and variational dynamics share the same initial condition: $\ket{\psi(0)} = \ket{\psi^{\rm VS}(0)} = \ket{\psi_0(0)}$. From now on, we introduce the notation $\ket{\psi^{\rm VS}_0(s)}$ for the instantaneous variational ground state.
Under certain hypotheses detailed in \aref{app: vqa},
the time-dependent variational state $\ket{\psi^{\rm VS}(s)}$ remains close to the instantaneous variational ground state, namely $\ket{\psi^{\rm VS}(s)}=\ket{\psi^{\rm VS}_0(s)}+\ket{\delta\psi(s)}$, where $\lVert\delta\psi(s)\rVert_2=\mathcal{O}(1/T)$, for $s \in [0,1]$.
If the target ground state belongs to the manifold, i.e.\ $\ket{\psi_0(1)}=\ket{\psi^{\rm VS}_0(1)}$, then the \emph{variational} state converges to the target ground state at $s=1$, for large values of $T$.
In practice, we shall consider a low-entanglement manifold and classical initial and final (target) states.

A more precise statement of the theorem is given in \aref{app: vqa}. Here, we provide only the physical intuition illustrated in \fref{fig: manifold volution}. If $\lVert\delta\psi(s)\rVert_2$ is small, its evolution is determined by the linearized time-dependent variational equations that give rise to elliptic trajectories around the instantaneous variational ground state $\ket{\psi^{\rm VS}_0(s)}$. The frequencies $\omega_l>0$ of the trajectories are determined by the eigenvalues of the linearized evolution map on the manifold $\mathcal{M}$~\cite{hackl2020geometry} and represent approximate energy gaps over the ground state. If $1/T\ll \omega_l$, then $\ket{\delta\psi(s)}$ will average to zero in a time window where $\ket{\psi^{\rm VS}_0(s)}$ does not change drastically. Therefore, $\ket{\psi^{\rm VS}(s)}$ will follow the instantaneous variational ground state $\ket{\psi^{\rm VS}_0(s)}$ and $\lVert\delta\psi(s)\rVert_2$ will remain small during the variational time evolution.

In the following, we consider several examples of increasing complexity, with protocols in the form
\begin{align}
\label{eq: protocol}
        H(s) &= (1-s)H_0+sH_1+s(1-s)H_2, \\ \nonumber
        H_0 & = -\sum_{j=1}^N\sx_j,
\end{align}
where we may introduce an additional \textit{catalyst} Hamiltonian $H_2$. In all cases, we initialize the system in the ground state of $H_0$ and evaluate the convergence of the variational evolution to the target ground state of $H_1$ in a total time $T$.

\paragraph{Two-qubit model.}
First, we consider a two-qubit model with
\begin{align}\label{eq:2-qubit H}
    H_1 & = \sz_1\sz_2-2(\sz_1+\sz_2), \\ \nonumber
    H_2 & = -2A \ket{\Phi}\bra{\Phi},
\end{align}
where $\ket{\Phi}=\frac{1}{\sqrt{2}}(\ket{00}+\ket{11})$ is a maximally entangled Bell state, and $A$ is a positive constant. The initial and the final ground states are product states where both spins point in the $x$ and $z$ directions, respectively. Therefore, we choose the following product state variational ansatz 
\begin{equation}
    \begin{aligned}
    \label{eq: 2-qubit psi}
    \ket{\psi^{\rm VS}_{\theta,\phi}} &= \ket{\psi_{\theta,\phi}}\otimes \ket{\psi_{\theta,\phi}}, \\ 
    \ket{\psi_{\theta,\phi}} &= \cos(\theta/2)\ket{0} + \sin(\theta/2)\mathrm{e}^{-{\im}\phi}\ket{1}.
    \end{aligned}    
\end{equation}
The variational manifold given by \eref{eq: 2-qubit psi} is a K\"ahler manifold (see \aref{app: 2-partite model} or \cite{hackl2020geometry}). The parameter $A$ determines the entanglement of the ground states of $H(s)$ at intermediate times $0<s<1$. In \fref{fig: 2-qubit scaling}, we show that cases with higher ground-state entanglement at $s=0.5$ correspond to faster convergence to the variational ground state at the end of the protocol (see \fref{fig: 2-qubit scaling}~a),~b)). 
This example demonstrates that the variational adiabatic evolution can still reach the target ground state, even if the exact quantum dynamics would access high-entanglement regions at intermediate times. 
\fref{fig: 2-qubit scaling}~c) shows the scaling of $\lVert \delta\psi(1)\rVert_2 = \mathcal{O}(1/T)$, which is verified regardless of the value of $A$, as predicted by the variational ground-state adiabatic theorem.   
This two-qubit model can be extended to a more general bipartite system involving arbitrarily high entanglement at the intermediate ground states of $H(s)$ (see \aref{app: 2-partite model}). Nonetheless, the variational evolution on a product-state manifold
is equivalent to the two-qubit case and correctly leads to the target ground state. 

\begin{figure}[!t]
    \centering
    \includegraphics[width=\columnwidth]{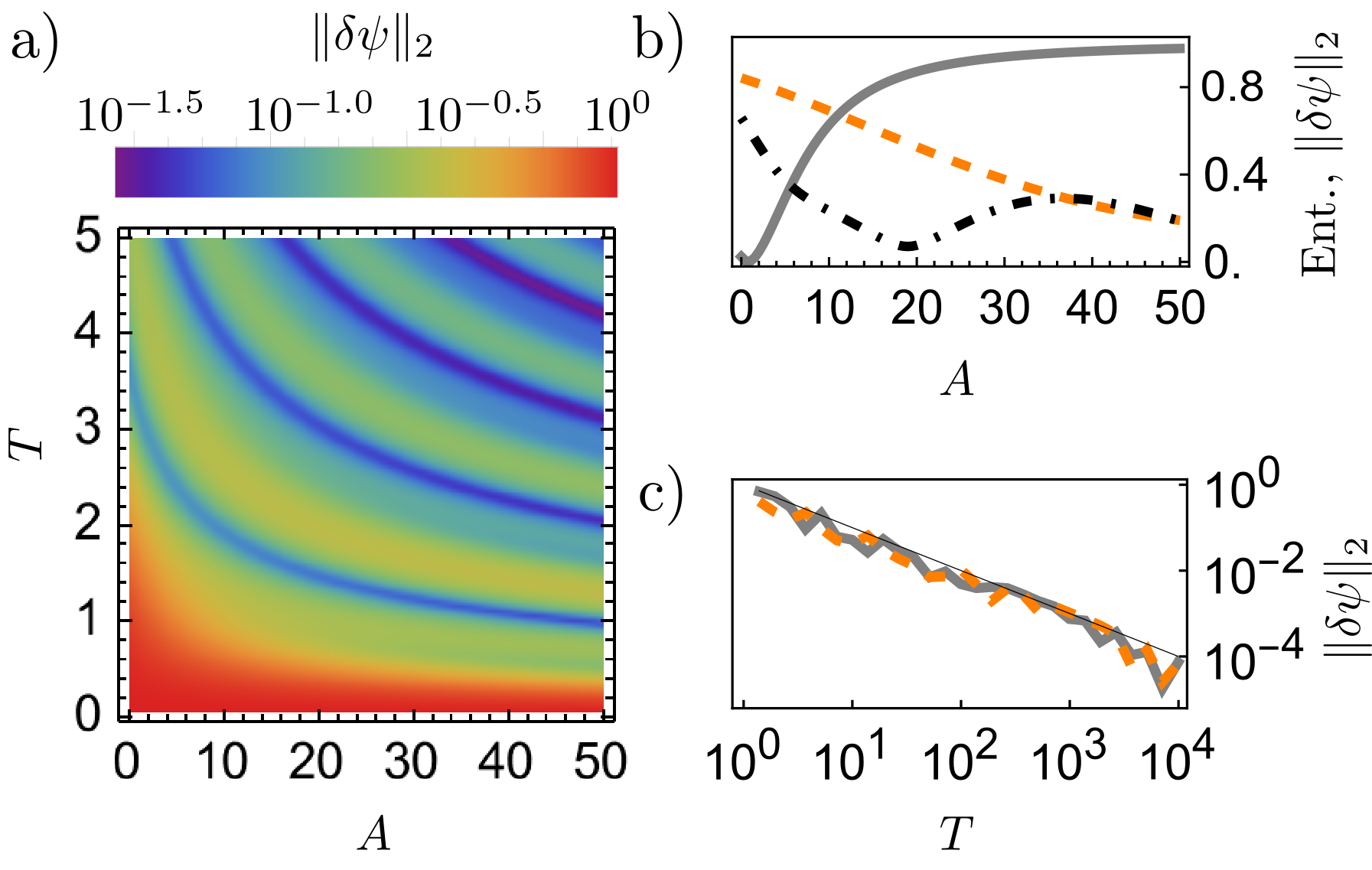}
    \caption{\textit{a)} The final norm $\lVert\delta\psi(s)\rVert_2$ at $s=1$, a measure of the distance between $\ket{\psi_0(1)}$ (the target ground state of $H_1$) and $\ket{\psi^{\rm VS}(1)}$ (the variational state at the end of the annealing). 
    We plot it as a function of the annealing time $T$ and parameter $A$. 
    \textit{b)} Half-system entanglement entropy at intermediate time $s=0.5$ (grey line) as a function of the parameter $A$, and the norm $\lVert \delta\psi(1)\rVert_2$ for $T=0.5$ --- dashed orange, and $T=1.5$ --- dash-dotted black. Higher entanglement entropy corresponds to smaller values of $\lVert \delta\psi (1)\rVert_2$. 
    \textit{c)} Scaling of the norm $\lVert \delta\psi(1)\rVert_2$ with annealing time $T$. The thin line corresponds to $T^{-1}$ scaling, the dashed orange line to $A=0$, and the grey line to $A=5$. 
    }
    \label{fig: 2-qubit scaling}
\end{figure}

\paragraph{Lipkin-Meshkov-Glick model.}
Next, we consider the integrable Lipkin-Meshkov-Glick (LMG)~\cite{lipkin1965validity} 
\begin{equation}
    H_1^{\text{LMG}} = -  \frac{1}{N} \Biggl(\sum_{i=1}^N \sigma_i^z \Biggr)^2 \, 
\end{equation}
as target Hamiltonian. Similarly to the two-qubit case, we have homogeneous product state ground states of the initial and final Hamiltonian pointing in the $x$ and $z$ directions, respectively. Therefore, we consider the following mean-field ansatz
\begin{equation}
    \ket{\psi^{\rm VS}_{\theta,\phi}}= \ket{\psi_{\theta,\phi}}^{\otimes N},
\end{equation}
with $\ket{\psi_{\theta,\phi}}$ defined as in \eref{eq: 2-qubit psi}. The variational manifold is again a K\"ahler manifold with simple equations of motion for the variational parameters that can be solved numerically for any $N$ (see \aref{app: LMG}). 
This manifold can be effectively represented as a sphere by considering the expectation value of total spin $S$, with components
\begin{align}
    S^\alpha = 
    \frac{1}{N} \sum_{i=1}^N
    \braket{\psi_{\theta,\phi}|\sigma_i^{\rm \alpha}|\psi_{\theta,\phi}}, \quad \alpha={\rm x,y,z}.
\end{align}
The LMG model crosses the second-order quantum phase transition at $s^*=1/3$ in the thermodynamic limit. In this limit, states on the manifold exactly describe the dynamics and some properties of the eigenstates~\cite{granet2023exact,vzunkovivc2023mean,vzunkovivc2024mean}, and the system behaves as a classical spin on the sphere~\cite{lipkin1965validity}.
At finite $N$, the instantaneous quantum many-body ground state does not belong to the manifold, except at $s=0,1$. 
However, the time-dependent variational state and the instantaneous variational ground state are still represented on the sphere. If the norm of their difference $  \lVert\delta S (s)\rVert_2 = \lVert S^{\rm\, VS}(s)-S_0^{\rm VS}(s) \rVert_2$ converges to small values at $s=1$, the variational evolution approximates the exact target ground state.

As shown in \aref{app: LMG}, the finite $N$ equations of motion on the manifold retain the same structure as in the thermodynamic limit and display a mean-field second-order phase transition at $s^*=\frac{N}{3N-2}$.
%
For $s<s^*$, the instantaneous variational ground state does not depend on $s$, hence, for any value of $T$, $\lVert\delta S (s)\rVert_2$ oscillates close to its initial value at $s=0$;
for $s>s^*$, the convergence of the time-dependent variational state to the instantaneous variational ground state scales as $ \lVert\delta S (s)\rVert_2 \sim 1/\sqrt{T}$.
Indeed, our assumptions in the variational ground-state quantum adiabatic theorem are not strictly valid in this regime, as $\lVert \partial_s S_0^{\rm VS}(s)\rVert$ exhibits a square root divergence at the critical time $s^*$, leading to a change in the exponent (see \aref{app: LMG}).

In \fref{fig: LMG}, we present a specific case with $N=4$.
Although we deliberately focus on a small system size, distant from the mean-field approximation, the same results apply to any value of $N$.
We start close to the $H_0$ ground state, with the initial condition shifted $10^{-4}$ away from the angles $\theta(0)=\pi/2$, $\phi_0=0$, which determine a fixed point for all $s$. In \fref{fig: LMG}~a) we show a trajectory of the vector $S^{\rm\, VS}(s)$ (black line) for $T=500$. We observe fast oscillations around the variational ground-state vector $S_0(s)$ (dashed orange line). These oscillations become faster and the amplitude smaller upon increasing $T$, as displayed in \fref{fig: LMG}~c), where we plot $\delta S(s)$ at $s \in [0,1]$ for three values of $T$. Consistently, \fref{fig: LMG}~b) shows the square root dependence of $\lVert \delta S\rVert_2$ on $T$, for $s>s^*$.
The jump at $s^*=0.4$ is related to the mean-field phase transition (see \aref{app: LMG}).
Finally, \fref{fig: LMG}~d) corroborates the agreement of numerical data with the square root scaling, verifying the convergence of the final variational state $(s=1)$ to the target ground state.
\begin{figure}[!htb]
    \centering
    \includegraphics[width=\columnwidth]{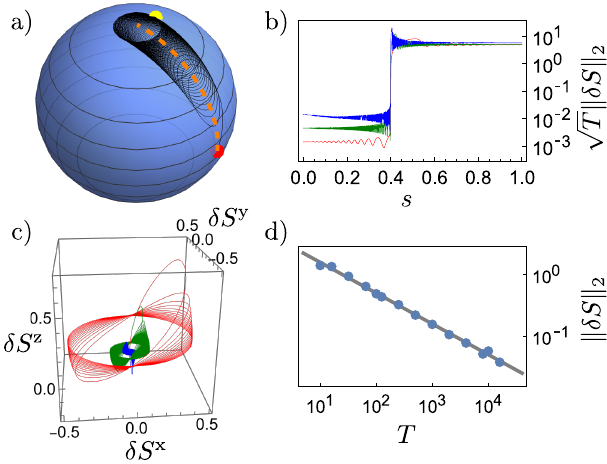}
    \caption{\textit{a)} An example evolution of the vector $S^{\rm\, VS}(s)$ (black line) for $N=4$ and $T=500$, $s \in [0,1]$. The red dot corresponds to the initial condition, and the yellow dot corresponds to the final variational state $S^{\rm\, VS}(1)$. The dashed orange line shows the \emph{variational} ground state vector $S_0^{\rm VS}(s)$. 
    \textit{b)} Three examples of the evolution of the rescaled norm $\sqrt{T}\lVert \delta S (s)\rVert_2$ for $T=10^2$ (red), $T=10^3$ (green), and  $T=10^4$ (blue).
    \textit{c)} Convergence of $\delta S (s)$ during the variational time evolution
    --- same data as in panel \textit{b)}.
    \textit{d)} Scaling of 
    $\lVert\delta S (1)\rVert_2$: the solid line represent the predicted $1/\sqrt{T}$, the dots correspond to numerical results.}
    \label{fig: LMG}
\end{figure}

\paragraph{Ising spin glass.} 
Finally, we consider the fully connected random Ising spin glass model with the Hamiltonian
\begin{equation}\label{eq: SG}
    H_1^{\rm SG} = - \sum_{i,j=1}^N J_{ij} \sigma_i^z \sigma_j^z + \sum_{i=1}^N h_i \sigma_i^z\,,
\end{equation}
where $J_{ij} = J_{ji}$ and $h_i$ are random numbers uniformly distributed in the range $(0,1)$ and $(-0.5,0.5)$, respectively. Determining the ground state of the model is computationally equivalent to quadratic unconstrained binary optimization (QUBO) and is an NP-complete problem~\cite{barahona1982computational}. 

Here, we use the matrix product state (MPS) variational ansatz, which enables a systematic increase of maximal entanglement by increasing the bond dimension $D$. The MPS manifold is a K\"ahler manifold~\cite{haegeman2014geometry} that enables an inverse-free implementation of the time-dependent variational principle (TDVP)~\cite{haegeman2016unifying} and a simple evaluation of the gap~(see \aref{app: spin glass}).
\begin{figure}[!htb]
    \centering
    \includegraphics[width=\columnwidth]{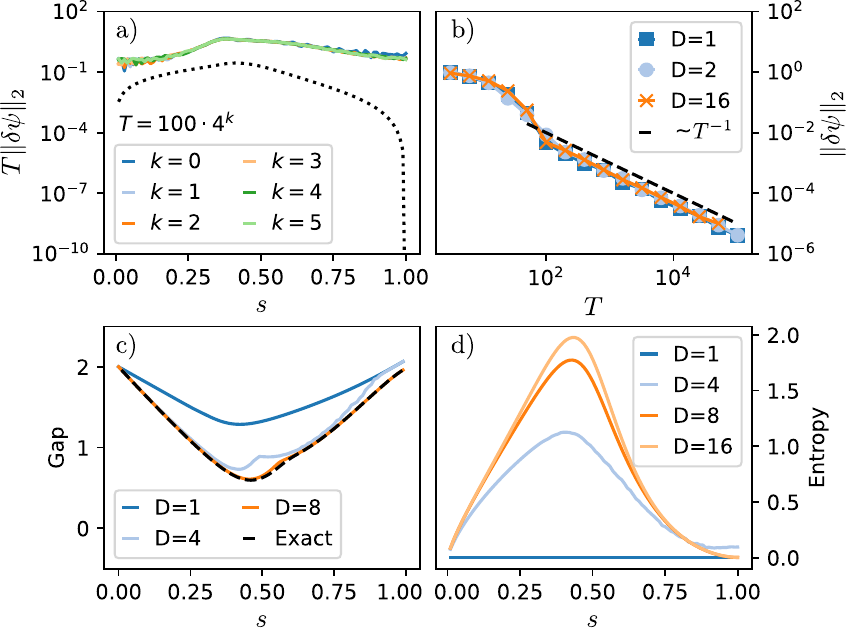}
    \caption{\textit{a)} Convergence of the time-dependent variational state towards the instantaneous variational ground state upon increasing the annealing time $T$ for a particular realization of the Ising spin glass model with $N=8$. We fix the bond dimension to $D=1$ and plot $\lVert\delta\psi(s)\rVert_2$. The black dotted line shows the norm of the difference between the exact instantaneous ground state and the instantaneous variational ground state. 
    \textit{b)} Norm of the difference between the final variational state at $s=1$, which is $\ket{\psi^{\rm VS}(1)}$, and the exact target ground state of $H_1$, i.e.\ $\ket{\psi_0(1)}$, for increasing values of $D$.
    Notice that the target state belongs to the MPS manifold with $D=1$ (dotted black line in panel~\textit{a}, at $s=1$).
    The bond dimension $D=16$ corresponds to the exact evolution, and the dashed line highlights the predicted scaling $1/T$.
    \textit{c)} The spectral gap calculated from the effective Hamiltonian in the MPS-TDVP evolution for various bond dimensions $D$. The dashed line corresponds to the exact spectral gap. We set $T=3200$.
    \textit{d)} Entanglement entropy during the protocol with $T=3200$ and different bond dimensions.
    }
    \label{fig: ising scaling}
\end{figure}

In \fref{fig: ising scaling}~a) we show that the time-dependent variational state $\ket{\psi^{\rm VS}}$ with $D=1$ converges uniformly, i.e.\ for all $0<s\leq1$, towards the instantaneous variational ground state with the same bond dimension. The expected $1/T$ convergence remains valid upon increasing the bond dimension, as shown in \fref{fig: ising scaling}~b). 
As predicted by our theorem, we obtain the target ground state even by constraining the variational adiabatic evolution to a low-entanglement manifold.
The results shown in \fref{fig: ising scaling} refer to a specific disordered instance of the Ising spin glass. Still, we verified their validity on different realizations of the random quenched disorder.

The MPS-TDVP approach enables the efficient estimate of the instantaneous gap through the effective Hamiltonian (see \aref{app: spin glass}).
%
In \fref{fig: ising scaling}~c), we show the convergence of the effective gap with increasing bond dimension. The exact gap at the protocol's beginning and end is well approximated, even with a small bond dimension.
This is expected since the initial and final (target) ground states are product states.
In contrast, a correct estimation in the middle of the protocol requires higher bond dimension values. This overall trend reflects the overlap between the exact and the variational ground states (see black dotted line in panel \textit{a}, for the case of $D=1$).
Consistently, ground-state entanglement is large in the middle of the protocol, as shown in \fref{fig: ising scaling}~d).

In \fref{fig: ising stat}, we present the histogram of the norm $\lVert\delta\psi(s)\rVert_2$ for $s=1$, representing the distance of the final variational state from the exact target ground state, for 100 random realizations of the Hamiltonian in \eref{eq: SG}.
We observe that for $D=1$, $D=2$, $D=4$, and $D=8$, approximately 20\%, 8\%, 2\%, and 1\% of the instances do not converge accurately to the target ground state for $T=1600$, i.e., the final norm $\lVert \delta\psi(1)\rVert_2>0.1$. 
This is related to spurious first-order phase transitions for the variational parameters, which appear more likely for smaller values of $D$, namely for lower-dimensional manifolds (see \aref{app: spin glass} for details).
However, in all cases, we were able to recover the correct ground state with an additional Density Matrix Renormalization Group (DMRG) run 
at the end of the protocol.
Even if the final variational state has a small overlap with the target state, it may serve as a warm-start for DMRG, leading to better convergence than random initializations.
\begin{figure}[!htb]
    \centering
    \includegraphics[width=0.9\columnwidth]{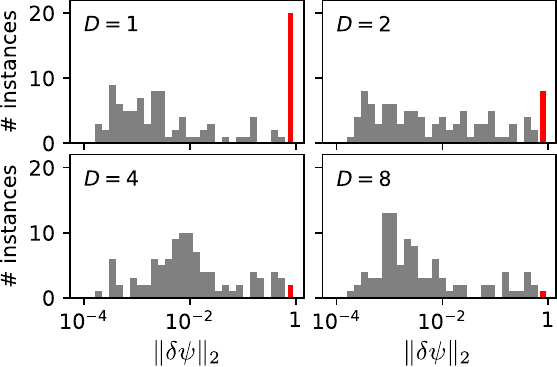}
    \caption{Histogram of the norm $\lVert\delta\psi(1)\rVert_2$ at the end of the annealing protocol with $D=1,2,4,8$ and $T=1600$.
    The histogram is computed from 100 random realizations of the final Hamiltonian $H_1$. We fixed $N=8$.
    The red bars correspond to instances with values larger than $0.1$. }
    \label{fig: ising stat}
\end{figure}

\paragraph{Conclusions and outlook.}

Our theorem states that, given a quantum adiabatic protocol with sufficiently slow evolution and a variational manifold, the time-dependent variational state remains close to the instantaneous variational ground state, provided the hypotheses are satisfied.
If the initial and final states lie within a classically simulable variational manifold, under the theorem's assumptions, this could open up a new promising computational paradigm:
an efficient quantum-inspired classical algorithm performs the manifold evolution as in \fref{fig: manifold volution}, finally reaching the target state.

This approach has been proven valid 
%
for low-entanglement variational manifolds and various classical target Hamiltonians, e.g.\ for a spin glass, resulting in numerical outcomes consistent with the predicted scaling.
Selecting an appropriate manifold, such as the MPS manifold with larger bond dimension, can enhance the final results and decrease the annealing time needed to find the target state.

The proposed variational adiabatic computing scheme represents a novel classical paradigm, distinct from stochastic annealing and simulated quantum annealing, which rely on Monte Carlo sampling. 
We anticipate that its ultimate effectiveness in solving complex classical optimization problems and other applications may be linked to two natural extensions of our work.
Firstly, our variational adiabatic strategy could be extended to imaginary or complex time. This would result in an effective cooling during the adiabatic evolution, which can further improve the results. Various combinations with DMRG are also possible.
Secondly, we plan to examine different variational manifolds, such as PEPS~\cite{Cirac_2021} or neural-network quantum states~\cite{Carleo_2017, Schmitt_2020, PhysRevB.106.L081111, VNA}.

These advances could enable tackling more challenging problems and comparing the performance on larger instances with state-of-the-art classical methods or experimental implementations of quantum annealing. Furthermore, the variational quantum adiabatic paradigm may be adapted to finding excited states of the energy spectrum, extending its application beyond ground-state search.
Some important practical benefits compared to a fully quantum adiabatic evolution are already evident, e.g.\ access to the exact variational quantum state, which enables boosting an exceedingly small probability of finding the correct solution (such as $\approx 10^{-12}$) to almost one: this can be accomplished via DMRG on the final variational state, or similar methods. 

From a theoretical perspective, the challenge of identifying a practical sufficient condition for the applicability of the variational adiabatic theorem remains open. 
Nonetheless, our methods could be readily applied to various contexts traditionally studied within the framework of quantum annealing. Examples include the generalized Landau-Zener problem~\cite{brundobler1993s}, mechanisms to circumvent slowdown for adiabatic quantum computation~\cite{albash2018adiabatic}, 
and the non-adiabatic evolution~\cite{crosson2014different}. Moreover, it would be possible 
to employ the manifold simulations to approximately estimate the position of the minimum adiabatic gap. 
Therefore, variational adiabatic evolution appears to be a novel and robust classical baseline for benchmarking quantum protocols and a valuable tool for developing novel classical algorithms inspired by quantum methods.

\paragraph{Acknowledgments.}
\begin{acknowledgments}
BZ was funded by ARIS project J1-2480 and ARIS research program P2-0209 Artificial Intelligence and Intelligent Systems (2022 – 2024).  GP acknowledges financial support from the QuantERA II Programme STAQS project that 
has received funding from the European Union’s H2020 research and innovation program under Grant Agreement No 101017733. GL was supported by ANR-22-CPJ1-0021-01 and ERC Starting Grant 101042293
(HEPIQ).
MC acknowledges support from the PNRR MUR project PE0000023-NQSTI and by the PRIN 2022 (2022R35ZBF) - PE2 - “ManyQLowD”. 
Computational resources were provided by SLING – a Slovenian national supercomputing network.
\end{acknowledgments}

\bibliography{sample}

\appendix

\section{Variational ground-state quantum adiabatic theorem}\label{app: vqa}
\begin{theorem}
Assume a time-dependent Hamiltonian $H(t)$ whose instantaneous ground state $\ket{\psi_0(t)}$ at time $t=0$ lies in the variational manifold $\mathcal{M}$. Introduce the rescaled time $s=t/T$, where $T$ is the annealing time determining the inverse rate of change of $H(t)$. 
Consider a variational state $\ket{\psi_x^{\rm VS}(s)}\in\mathcal{M}$ whose time-dependence is fully determined by the variational parameters $x=x(s)$. Starting in the initial ground state $\ket{\psi_x^{\rm VS}(0)}=\ket{\psi_0(0)}$, the variational parameters are evolved with the (Lagrangian action) variational principle
\begin{equation}
    \frac{1}{T}\frac{\rm d}{{\rm d}s}x  = \mathcal{X}, \quad \mathcal{X}= P_{x}(-\ii H(s)\left|\psi_x^{\rm VS}(s)\right\rangle)\,,
    \label{eq:vp action}
\end{equation}
where $P_{x}$ is the projector on the tangent manifold at $\ket{\psi_x^{\rm VS}(s)}$.
The resulting dynamics on the manifold can be decomposed as $\ket{\psi_x^{\rm VS}(s)}=\ket{\psi_{x_0}^{\rm VS}(s)}+\ket{\delta\psi(s)}$, where $\ket{\psi_{x_0}^{\rm VS}(s)}$ denotes the instantaneous variational ground state characterized by $x_0=x_0(s)$ .
If the linearized symplectic generator $K(s)$ defined by 
\begin{align}
K^\mu_{~\nu}(s)= \partial_{x^{\nu}}\mathcal{X}^\mu\Big\lvert_{x=x_0(s)},
\end{align}
satisfies the adiabatic conditions for $\eta$-pseudo-Hermitian quantum evolution, then 
\begin{align}
\lVert\delta\psi(s)\rVert_2=\mathcal{O}\left(\frac{1}{T}\right) \qquad \text{for} \qquad s \in [0,1]\,.
\end{align}
More precisely, if $\delta x(s) = x(s) - x_0(s)$, we obtain
\begin{align}
\lVert\delta x (s)\rVert_{\eta(s)} \sim \kappa \left(\frac{1}{T}\right) \qquad \text{for} \qquad s \in [0,1]\,,
\end{align}
where the norm is calculated with respect to the pseudo-inner product and
\begin{align}
    \label{eq: app bound}
    \kappa = \frac{2\max_{s,s',j} \left| {\rm e}^{-\ii\big(\Gamma_{j} (s)-\Gamma_{j} (s')\big)} \right|\max_s\lVert\dot{x}_0(s) \rVert_{\eta(s)}}{\min_{s,i}|\omega_i(s)|},
\end{align}
with $\Gamma_j(s)$ denoting the generalized geometric phase. 
If the previous term does not diverge, the time-dependent variational state closely follows the instantaneous variational ground state.
\end{theorem}

\begin{proof}
Let us consider the time evolution of the difference between the parameters of the time-dependent variational state and those of the instantaneous variational ground state
\begin{align}
\label{eq: proof_delta_x}
    \delta x(s) = x(s) - x_0(s).
\end{align}
The equations of motion for $\delta x(s)$ are
\begin{align} \label{eq: tdvp chi}
    \frac{1}{T}\partial_s\left(x_0(s) + \delta x(s)\right) = \mathcal{X}\big(x_0(s)+\delta x(s)\big),
\end{align}
Initially, by hypothesis, we have $\delta x(0)=0$. 
Our strategy is to show that the pseudo-norm of the difference $\lVert\delta x (s)\rVert_{\eta(s)} $ remains small during the adiabatic evolution. 

We start by linearizing the dynamics in \eref{eq: tdvp chi} in a local neighborhood of $x_0(s)$ 
\begin{align}
    \frac{1}{T}\partial_s\left(x_0(s) + \delta x(s)\right) \approx \mathcal{X}(x_0(s)) + K(s) \delta x(s),
\end{align}
where $K^{\mu}_{~\nu}(s) = \partial_{x^\nu}\mathcal{X}^\mu\big|_{x=x_0(s)}$. Since the instantaneous variational ground state is a stationary point of the Lagrangian action variational evolution\cite{hackl2020geometry}, we have $\mathcal{X}(x_0(s))=0$ and 
\begin{align}
    \label{eq: ddx/ds}
    \partial_s \delta x(s) = T K(s) \delta x (s) - \partial_s x_0(s).
\end{align}
We find the solution of the above inhomogeneous first-order differential equation in terms of the fundamental matrix satisfying~\cite{notesode}
\begin{align}
    \label{eq: K diff eq}
    \partial_s U(s) = TK(s)U(s).
\end{align}
A formal solution to \eref{eq: K diff eq}  is determined by a time-ordered exponential 
\begin{align}
    U(s) = \mathcal{T}_{\leftarrow}{\rm e}^{\int_0^s T K(u){\rm d} u}.
\end{align}
We write $\delta x(s)$ as a sum of the homogeneous and the particular solution of \eref{eq: ddx/ds}. Due to the initial condition $\delta x(s=0)=0$, the homogeneous part of the solution vanishes since $\delta x _{\rm h}(s)=U(s)\delta x (0)=0$. The particular part, therefore, determines the final solution
\begin{align}
    \label{eq: delta x sol}
    \delta x (s)=\delta x _{\rm P}(s) = U(s)\int_0^s U^{-1}(u)\dot{x}_0(u){\rm d}u\,,
\end{align}
where we introduced the notation
$\dot{x}_0(s)=\partial_s x_0(s)$

To further simplify the expression \eref{eq: delta x sol}, we use the structure of the symplectic generator $K$, which is always diagonalizable and whose eigenvalues appear in conjugate imaginary pairs~\cite{hackl2020geometry}. Consequently, ${\rm i} K$ has real eigenvalues implying the existence of a strictly-positive definite metric operator $\eta$ such that \cite{mostafazadeh2010pseudo,mostafazadeh_pseudo-hermiticity_2002}:
\begin{align}
({\rm i}K)^\dag\eta=\eta ({\rm i}K).
\end{align}
The above relation implies that ${\rm i}K$ is Hermitian with respect to the pseudo-inner product 
\begin{align}\label{eq: pseudo product}
\braket{\psi|\phi}_\eta = \bra{\psi}\eta\ket{\phi}\,.
\end{align}
Therefore, \eref{eq: K diff eq} is a real-time Schr\"odinger equation with a pseudo-Hermitian Hamiltonian ${\rm i}K$~\cite{mostafazadeh2010pseudo, brody2013biorthogonal, mostafazadeh2020time}.
Although we shall adopt the standard Dirac notation, notice that the operator ${\rm i}K$ and the parameter vectors introduced in \eqref{eq: proof_delta_x} act on the variational manifold $\mathcal{M}$.
We can find the left $\{\ket{\tilde{\omega}_i(s)}\}$ and right $\{\ket{\omega_i(s)}\}$ eigenvectors of ${\rm i}K(s)$ with the corresponding eigenvalues $\omega_i(s)$, namely ${-\rm i}K(s)^\dag\ket{\tilde{\omega}_i(s)}=\omega_i(s)\ket{\tilde{\omega}_i(s)}$ and 
${\rm i}K(s)\ket{\omega_i(s)}=\omega_i(s)\ket{\omega_i(s)}$. The left and right eigenvectors are related via 
\begin{align}
\ket{\tilde{\omega}_i(s)}=\eta(s)\ket{\omega_i(s)}
\end{align}
and satisfy the bi-orthonormality condition 
\begin{align}
\braket{\omega_i(s)|\tilde{\omega}_j(s)}&=\braket{\omega_i(s)|\eta(s)|\omega_j(s)}\\ \nonumber&=\braket{\omega_i(s)|\omega_j(s)}_{\eta(s)}=\delta_{i,j}.
\end{align}

Recently, a pseudo-Hermitian generalization of quantum mechanics has been worked out~\cite{brody2013biorthogonal,mostafazadeh2010pseudo}. Since our linearized homogeneous equation \eref{eq: K diff eq} is exactly a Schr\"odinger equation for the pseudo-Hermitian Hamiltonian ${\rm i}K$ we apply the adiabatic theorem for pseudo-Hermitian quantum dynamics~\cite{nenciu1992adiabatic,mostafazadeh2010pseudo, mostafazadeh2020time, silberstein2020berry}. 
We assume 
\begin{align}
\frac{1}{T}\max_{s\in \left[0,1\right]}  \frac{\left|\bra{\tilde{\omega}_i(s)} \partial_s K(s)\ket{\omega_j(s)}\right|}{ 
\left| \omega_i(s) - \omega_j(s) \right|^2} \ll 1,\quad \forall i\neq j.
\end{align}
According to the pseudo-Hermitian quantum adiabatic theorem~\cite{mostafazadeh2020time, silberstein2020berry}, we have
\begin{align}
\label{eq:adiabatic K}
    U(s)\ket{\omega_j(0)}\approx\,& {\rm e}^{\ii T \Phi_{j}(s)} \ {\rm e}^{\ii \Gamma_{j} (s)} \ket{\omega_j(s)},\\ \nonumber
    U^{-1}(s)\ket{\omega_j(s)}\approx\,& {\rm e}^{-\ii T \Phi_{j}(s)} \ {\rm e}^{-\ii \Gamma_{j} (s)} \ket{\omega_j(0)},
\end{align}
where 
\begin{align*}
\Phi_{j}(s) =& \int_0^{s}\omega_j(u){\rm d} u, \\
\Gamma_{j}(s) =&\frac{1}{2}\int_0^{s}\left(\braket{\dot{\tilde{\omega}}_j(u)|\omega_j(u)} - \braket{\omega_j(u)|\dot{\tilde{\omega}}_j(u)}\right){\rm d}u .
\end{align*}
The key difference from the standard Hermitian adiabatic theorem is that the geometric phase $\Gamma_{j} (s)$ does not vanish for open paths and is not necessarily real. Consequently, it can contribute to the global phase and change the norm of the eigenstates. However, the generalized geometric phase does not depend on time, hence it contributes only a prefactor to the final result. 

Since the left and right eigenvectors satisfy the completeness relation $\mathds{1}=\sum_j \ket{\omega_j(s)}\bra{\tilde{\omega}_j(s)}$, we can expand the inhomogeneous vector term as
\begin{align}
\label{eq:x0 expansion}\ket{\dot{x}_0(s)}=\sum_{j}\ket{\omega_j(s)}\braket{\tilde{\omega}_j(s)|\dot{x}_0(s)}\,,
\end{align}
where we exploit the bra-ket notation.
After inserting the expansion \eref{eq:x0 expansion} into the expression in \eref{eq: delta x sol} and using the pseudo-Hermitian adiabatic approximation, we find
\begin{align}
    \ket{\delta x(s)}& \approx \sum_j\int_0^{s}{\rm e}^{{\rm i}T(\Phi_j(s)-\Phi_j(u))} g_j(s,u) {\rm d}u \, \ket{\omega_j(s)},
\end{align}
where
\begin{align}
g_j(s,u) = {\rm e}^{\ii( \Gamma_{j} (s)-\Gamma_{j} (u))}\braket{\tilde{\omega}_j(u)|\dot{x}_0(u)}.
\end{align}

We aim to bound the norm $\lVert\delta x(s)\rVert_\eta$. Due to the bi-orthogonality relation, by using the pseudo-inner product definition in \eref{eq: pseudo product}, we obtain
\begin{align}\label{eq: dx adiabatic}
    \braket{\delta x(s)|\delta x(s)}_{\eta(s)} \approx 
    \sum_j\left|\int_0^{s}{\rm e}^{-{\rm i}T\Phi_j(u)}g_j(s,u) {\rm d}u \right|^2 \,.
\end{align}
We assume that $\omega_i(s)$, which represent the approximations of the excitation energies over the ground state, do not vanish. Now we can evaluate the integral in \eref{eq: dx adiabatic} using integration by parts and neglecting the second part since it is asymptotically smaller for large $T$
\begin{align*}
    \lVert\delta x (s)\rVert_{\eta(s)}^2 \approx \sum_{j}\left\lvert \frac{\ii}{T}\frac{g_j(s,u)}{\omega_j(u)}{\rm e}^{-\ii T \int_0^{u}\omega_j(u'){\rm d} u'}\Big\rvert_{u=0}^{u=s}\right\rvert^2.
\end{align*}
After some algebra, we find the following upper bound for $\lVert\delta x (s)\rVert_{\eta(s)}^2$:
\begin{equation}
   \left(\frac{2\max_{s,s',j} \left| {\rm e}^{-\ii( \Gamma_{j} (s)-\Gamma_{j} (s'))} \right|\max_s\lVert\dot{x}_0(s) \rVert_{\eta(s)}}{T\min_{s,i}|\omega_i(s)|}\right)^2.
\end{equation}

In cases where $\omega_i(s)$ vanishes, we would have a degenerate ground state, and the assumptions of the pseudo-Hermitian adiabatic theorem would not be valid.

In conclusion, let us write the time-dependent variational state as
$\ket{\psi_x^{\rm VS}(s)}=\ket{\psi_{x_0}^{\rm VS}(s)}+\ket{\delta\psi(s)}\,,$
where $\ket{\psi^{\rm VS}_{x_0}(s)}$ is the  variational ground-state.
If the variational parameters $x$ are close to the ground-state variational parameters $x_0$, we can perform a first-order expansion centered on $x_0$, leading to
\begin{align}
\ket{\delta \psi} =\sum_\mu \partial_{x^\mu} \ket{\psi_x^{\rm VS}}|_{x=x_0} \delta x_j\,,
\end{align}
where we dropped the $s$ dependence for conciseness.
The inner product of partial derivatives defines the real-valued metric $\bf g$ and the symplectic form $\boldsymbol{\omega}$ on the variational manifold~\cite{hackl2020geometry}
\begin{align}
\braket{\partial_{x^\mu}\psi_x^{\rm VS}|\partial_{x^\nu}\psi_x^{\rm VS}} = \frac{1}{2} (g_{\mu,\nu} + {\rm i}\omega_{\mu,\nu}).
\end{align}
Finally, by restoring the $s$ dependence, we have
\begin{equation}
    \begin{aligned}
    \lVert\delta\psi(s)\rVert_2 ^2 &=\frac{1}{2}\,x(s)^T \big({\bf g}(s)+{\rm i}\boldsymbol{\omega}(s)\big)x(s) \\ 
    &\leq  \frac{\vert \lambda_{\rm max}(s) \vert}{2} \lVert\delta x(s)\rVert_2^2 \\
    &\leq \frac{\vert \lambda_{\rm max}(s) \vert}{ 2\Lambda_{\rm min}(s)}  \lVert\delta x (s)\rVert_{\eta(s)}^2 ,
\end{aligned}
\end{equation}
where $\lambda_{\rm max}$ is the eigenvalue with the maximum modulus of the Hermitian matrix $({\bf g}+{\rm i}\boldsymbol{\omega})$, and $\Lambda_{\rm min}$ is the minimum eigenvalue of the positive definite metric operator $\eta$.
If the assumptions of the theorem are verified, this leads to $\lVert\delta\psi(s)\rVert_2=\mathcal{O}(1/T)$, for $s \in [0,1]$.

\end{proof}

\section{Maximally entangled model}\label{app: 2-partite model}
In this section, we generalize the two-qubit model discussed in the main text.
Let us consider a bipartite system where each subsystem $A$ and $B$ have a local 
Hilbert space $\mathds{C}^{N+2}$, with a local basis $\ket{i}$, $i=0,1,\ldots, N+1$. 
The annealing Hamiltonian of the composite model consists of three terms $H(s)=(1-s)H_0+sH_1+ s(1-s)H_2$, with $s=t/\tau$ as usual. The first and the second terms act nontrivially only on the first two states $\{\ket{0}, \ket{1}\}$ of each subsystem:
\begin{equation*}
    \begin{aligned}
        H_0 =& -(\ket{0}\bra{1}+\ket{1}\bra{0})_A\otimes\mathds{1}_B-\mathds{1}_A\otimes(\ket{0}\bra{1}+\ket{1}\bra{0})_B \\
    H_1 =& (\ket{0}\bra{0}-\ket{1}\bra{1})_A\otimes (\ket{0}\bra{0}-\ket{1}\bra{1})_B - \\
    &2\big[(\ket{0}\bra{0}-\ket{1}\bra{1})_A\otimes\mathds{1}_B+\mathds{1}_A\otimes(\ket{0}\bra{0}-\ket{1}\bra{1})\big]_B
    \end{aligned}
\end{equation*}
The term $H_2$ is a projector on a maximally entangled state $\ket{\Psi^+}=\frac{1}{\sqrt{N+2}}\sum_{i=0}^{N+1}\ket{i}$, namely
\begin{align}
    H_2 = - A(N+2)\ket{\Psi^+}\bra{\Psi^+},
\end{align}
where $A$ is a positive constant.
Notice how these operators act equivalently as in the two-qubit case upon restricting to the subspace spanned by $\{\ket{0}, \ket{1}\}$ of each subsystem $A$, $B$.
Moreover, if $N=0$, this model corresponds exactly to the two-qubit model discussed in the main text. 
We can consider a variational ansatz restricted to the relevant two-state subspaces mentioned above:
\begin{equation}
    \label{eq: app me ansatz}
    \begin{aligned}
    \ket{\psi_{\theta,\phi}^{\rm VS}} &= \ket{\psi_{\theta,\phi}}_A\otimes \ket{\psi_{\theta,\phi}}_B \\ 
    \ket{\psi_{\theta,\phi}} &= \cos(\theta/2)\ket{0} + \sin(\theta/2)\mathrm{e}^{-i\phi}\ket{1}.
\end{aligned}
\end{equation}
The tangent space basis vectors at $(\theta,\phi)$ are given by
\begin{align}    
\label{eq:app tangent V}
\ket{V_j}=&\ket{v_j}\otimes\ket{\psi_{\theta,\phi}}+\ket{\psi_{\theta,\phi}}\otimes\ket{v_j},\quad j=1,2\\ \nonumber
\ket{v_1} =& \frac{1}{2}\left(-\sin \left(\theta/2\right)\ket{0}+e^{\im \phi } \cos
   \left(\theta/2\right)\ket{1}\right), \\ \nonumber
   \ket{v_2} = & \frac{\rm i}{2}\left(-\sin \left(\theta/2\right) \sin (\theta )\ket{0}+ e^{\im
   \phi } \cos \left(\theta/2\right) \sin (\theta ) \ket{1}\right).
\end{align}
The metric $\mathbf{g}_{i,j}=2{\rm Re}\braket{V_i|V_j}$ and the symplectic form $\mathbf{\omega}_{i,j}=2{\rm Im}\braket{V_i|V_j}$ on the manifold are then given by (see also Ref \cite{hackl2020geometry})
\begin{align}
    \mathbf{g}=\left(\begin{matrix}
        1&0\\
        0&\sin\theta
    \end{matrix}\right), \quad
    \mathbf{\omega}=\left(\begin{matrix}
        0&\sin\theta\\
        -\sin\theta&0
    \end{matrix}\right).
\end{align}
By defining $\mathbf{G}=\mathbf{g}^{-1}$ as the inverse of the metric, it is easy to show that $\mathbf{J}=-\mathbf{G}\mathbf{\omega}$ is a complex structure satisfying $\mathbf{J}^2=-\mathds{1}$. Hence, the variational manifold given by the ansatz in \eref{eq: app me ansatz} is a K\"ahler manifold, which leads to the equivalence between different variational principles~\cite{hackl2020geometry}. 

Given the above definitions, we can calculate the equations of motion for the variational parameters $x=(\theta,\phi)$  
\begin{align}\label{eq:app variational eq}
    \frac{{\rm d}x_i}{{\rm d}s}=-2 T \sum_{j=1,2}\mathbf{G}_{i,j}{\rm Re}\braket{V_j|{ \im}H(s)|\psi_{\theta,\phi}^{\rm VS}},
\end{align}
which simplify to
\begin{align*}\nonumber
    \frac{{\rm d}\theta}{{\rm d}s} = & -2 (s-1) T \sin (\phi (s)) \Big(A\, s \sin (\theta (s)) \cos (\phi (s))+1\Big),\\ 
    \frac{{\rm d}\phi}{{\rm d}s} = & T \Big(s \cos (\theta (s)) \big(-A (s-1) \cos (2 \phi (s)) + \\ \nonumber& A (s-1)+2\big)+2 h s-2 (s-1) \cot (\theta (s))
   \cos (\phi (s))\Big).
\end{align*}
The above equations do not depend on the local Hilbert space dimension $N+2$. Therefore, we can engineer an arbitrarily high entanglement of the ground state at $s=1/2$ by choosing a large value for the parameter $A$. Still, our variational evolution constrained to a low-entanglement manifold (here, with zero entanglement) follows the instantaneous variational ground state. Since the target ground state of $H_1$ has zero entanglement, the variational evolution successfully converges to it at the end of the protocol. 

These results are shown in \fref{fig: nbit A}~a), where we plot the norm $\lVert\delta\psi(s)\rVert_2$ for $s=1$ (solid black line), a measure of the distance between $\ket{\psi_0(1)}$ (the target ground state of $H_1$) and $\ket{\psi^{\rm VS}(1)}$ (the variational state at the end of the annealing), as a function of $A$. 
In the same panel, the dashed orange line represents the norm of the difference between the exact and the variational ground state in the middle of the protocol, i.e.\ 
$\lVert \ket{\psi^{\rm VS}_0(s)}-\ket{\psi_0(s)} \rVert_2$ for $s=0.5$. 
In panel~b), we see the entanglement entropy of the exact ground state $\ket{\psi_0(s)}$, also at intermediate $s=0.5$.
The qualitative picture is consistent with the discussion presented in the main text for the two-qubit model.
Increasing values of $A$ correspond to higher entanglement entropy generation during the full adiabatic evolution.
This leads to a larger difference between the exact and the variational instantaneous ground states at intermediate times.
Nonetheless, the variational adiabatic evolution on the product-state manifold successfully converges to the target ground state. The final values of $\lVert\delta\psi(1)\rVert_2$ could be lowered by increasing the annealing time $T$.
\begin{figure}[!htb]
    \centering
    \includegraphics[width=0.7\columnwidth]{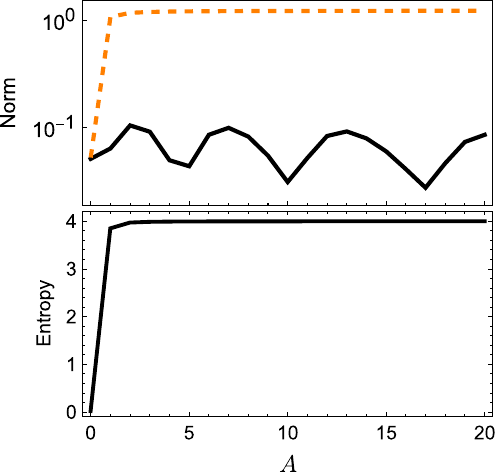}
    \caption{\textit{a)} 
    The final norm $\lVert\delta\psi(s)\rVert_2$ at $s=1$ (solid black line)
    and $\lVert \ket{\psi^{\rm VS}_0(s)}-\ket{\psi_0(s)} \rVert_2$ for $s=0.5$ (dashed orange line), as a function of $A$. \textit{b)} Half-system entanglement entropy of the exact ground state in the middle of the protocol ($s=0.5$). The annealing time is fixed to $T=10$ and $N=14$.}
    \label{fig: nbit A}
\end{figure}

\section{LMG model}\label{app: LMG}
In this section, we discuss the Lipkin-Meshkov-Glick model, the properties of the variational manifold, and derive the variational equations of motion.

The annealing Hamiltonian is $H(s)=(1-s)H_0+sH_1$ with
\begin{align}
    H_0 & = -\sum_{j=1}^N\sx_j, \\\nonumber
    H_1^{\text{LMG}} &= - \frac{1}{N} \Biggl(\sum_{i=1}^N \sigma_i^z \Biggr)^2.
\end{align}
As reported in the manuscript, we consider a homogeneous product-state variational manifold defined by the ansatz
\begin{align}
\label{eq: var_state_LMG}
    \ket{\psi^{\rm VS}_{\theta,\phi}}= \ket{\psi_{\theta,\phi}}^{\otimes N},
\end{align}
where $\ket{\psi_{\theta,\phi}}$ is defined in \eref{eq: 2-qubit psi}. The tangent vectors $\ket{V_{1,2}}$ are a generalisation of \eref{eq:app tangent V}, namely
\begin{align}
    \ket{V_j}=&\sum_{k=1}^N\ket{\tilde{V}_j^k},\\ 
\end{align}
where $\ket{\tilde{V}_j^k}$ is obtained from $\psi^{\rm VS}_{\theta,\phi}$ in \eqref{eq: var_state_LMG} by replacing the $k-$th vector $\ket{\psi_{\theta,\phi}}$ in the tensor product by $\ket{v_j}$. The metric $\mathbf{g}_{i,j}=2{\rm Re}\braket{V_i|V_j}$ and the symplectic form $\mathbf{\omega}_{i,j}=2{\rm Im}\braket{V_i|V_j}$ on the manifold are then given by 
\begin{align}
    \mathbf{g}=\frac{N}{2}\left(
\begin{array}{cc}
 1 & 0 \\
 0 & \sin ^2(\theta ) \\
\end{array}
\right),
\quad \mathbf{\omega}=\frac{N}{2}\left(
\begin{array}{cc}
 0 & \sin (\theta ) \\
 \sin(\theta)& 0 \\
\end{array}
\right).
\end{align}
The manifold is a K\"ahler manifold, since $\mathbf{J}=-\mathbf{G}\mathbf{\omega}$ is a complex structure satisfying $\mathbf{J}^2=-\mathds{1}$, where $\mathbf{G}$ is again the inverse of the metric, $\mathbf{G}=\mathbf{g}^{-1}$. The variational equation of motion \eref{eq:app variational eq} for $x=(\theta,\phi)$ simplify to
\begin{align}
    \frac{{\rm d} \theta}{{\rm d}s} =& 2 T (s-1) \sin (\phi (s)), \\ \nonumber
    \frac{{\rm d} \theta}{{\rm d}s} =& 2T \cos (\theta (s))\left(
    \frac{2 (n-1) s}{n}+\frac{(s-1) \cos (\phi (s))}{\sin (\theta (s))}
    \right).
\end{align}
The variational manifold can be conveniently represented with a Bloch sphere vector $S$, with 
\begin{align}
\label{eq: order parameter S}
    S^\alpha = 
    \frac{1}{N} \sum_{i=1}^N
    \braket{\psi_{\theta,\phi}|\sigma_i^{\rm \alpha}|\psi_{\theta,\phi}}, \quad \alpha={\rm x,y,z}.
\end{align}
The variational energy $E^{\rm VS}(\theta,\phi)=\braket{\psi^{\rm VS}_{\theta,\phi}|H(s)|\psi^{\rm VS}_{\theta,\phi}}$ is given by
\begin{align}
    E^{\rm VS}(\theta,\phi)=&-n (1-s) \sin (\theta ) \cos (\phi )\\ \nonumber
    &+(1-n) s \cos ^2(\theta )-1.
\end{align}
This system undergoes a mean-field phase transition at $s^*=\frac{n}{3 n-2}$ with the ground state parameters for $s>s^*$ given by
\begin{align}
    \theta_0(s)=\arcsin \left(\frac{n-n s}{2 s-2 n s}\right)+\pi,\quad \phi_0(s)=0\,
\end{align}
whereas for $s<s*$ the ground state is a paramagnetic state $\theta_0(s)=\pi/2$, $\phi_0(s)=0$.

The order parameter $S^z$ defined in 
\eqref{eq: order parameter S}, 
displays a second order phase transition at $s=s^*$ with the mean-field critical exponent $1/2$, namely $S^{\rm z}(s) = |s-s^*|^{1/2}$, leading to a square root divergence $\lVert \partial_s S_0(s)\rVert_2 \propto |s-s^*|^{-1/2}$. Approximating $s-s^*\approx \Delta t/T$, where $\Delta t$ is the integration time step (independent of $T$), and using the bound \eref{eq: app bound}, we find
\begin{align}
    \lVert\delta S(s>s^*)\rVert_2=O\left(\frac{1}{\sqrt{T}}\right). 
\end{align}

\section{Spin glass model}\label{app: spin glass}

\begin{figure}[H]
    \centering
    \includegraphics[width=\columnwidth]{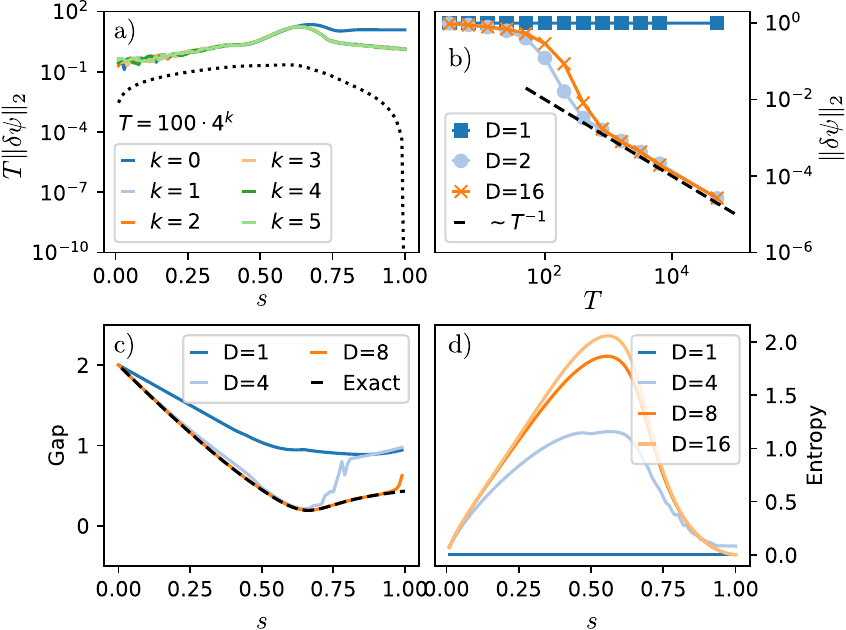}
    \caption{ 
    Convergence of the time-dependent variational state towards the variational ground state upon increasing the annealing time $T$ for a particular realization of the Ising spin glass model with $N=8$. \textit{a)} We fix the bond dimension to $D=2$ and show convergence of the norm $\lVert \delta\psi(s)\rVert_2$. The black dotted line corresponds to the norm of the difference between the exact ground state and the variational ground state.
    \textit{b)} 
    The plot of $\lVert \delta\psi(1)\rVert_2$, the norm of the difference between the final variational state and the target ground state, for increasing values of $D$.
    The classical target state belongs to the MPS manifold with $D=1$.
    The bond dimension $D=16$ corresponds to the exact evolution, and the dashed line highlights the predicted scaling $1/T$.
    \textit{c)} The spectral gap calculated from the effective Hamiltonian in the MPS-TDVP evolution for various bond dimensions $D$. The dashed line corresponds to the exact spectral gap. We set $T=3200$.
    \textit{d)} Entanglement entropy during the protocol with $T=3200$ and different bond dimensions.
    }
    \label{fig: app sg scaling 1}
\end{figure}

To address the spin glass model discussed in the main text, we utilize the matrix product state (MPS) variational manifold, which is a K\"ahler manifold~\cite{haegeman2014geometry}.
We use the inverse-free formulation of the Lagrangian action (or time-dependent) variational principle~\cite{haegeman2016unifying}. This iterative approach relies on calculating, for each site $j$, the effective local Hamiltonian $H_j^{\rm eff}$. The latter is defined by contracting the Hamiltonian matrix product operator with two copies of the MPS on all sites but site $j$. The local MPS tensor $A_j$ is then evolved by $\dot{A}_j=-{\rm i} H^{\rm eff}_jA_j$. 
The spectrum of the local effective Hamiltonian $H^{\rm eff}_j$ is used to estimate the low-lying energy eigenvalues of the full Hamiltonian. 
We remark that the eigenvalues obtained from the effective Hamiltonian $H^{\rm eff}_j$ are different from those obtained from the linearized variational map ${\rm i}K$ but are easier to calculate. Both approaches approximate the energy gap well if the variational ground state is close to the exact ground state.

Next, we present additional numerical results for the spin glass model. 
We first focus on two instances where the final variational state does not converge to the target ground state at $s=1$. In the first example, this happens only for bond dimension $D=1$ (see \fref{fig: app sg scaling 1}), while for $D\geq2$ 
the expected convergence is recovered. We can explain this by spurious first-order phase transitions appearing in low-dimensional manifolds. 

In \fref{fig: seed 11 theta}~\textit{a}), we plot one of the variational parameters as a function of the rescaled time $s=t/T$. The orange dashed line represents the variational parameter during the adiabatic evolution on the manifold with $D=1$. In contrast, the black dotted line shows the value of the same parameter for the instantaneous variational ground state. We see a discontinuity at a critical $s_c\approx0.47$, after which the two values no longer coincide. Consistently, the time-dependent variational energy shown in \fref{fig: seed 11 theta}~\textit{b)} (orange solid line) follows the variational ground-state energy (black dashed line) up to $s_c$ and starts to separate afterward. At larger bond dimensions $D>1$, this first-order transition disappears. 

\begin{figure}[H]
    \centering
    \includegraphics[width=\columnwidth]{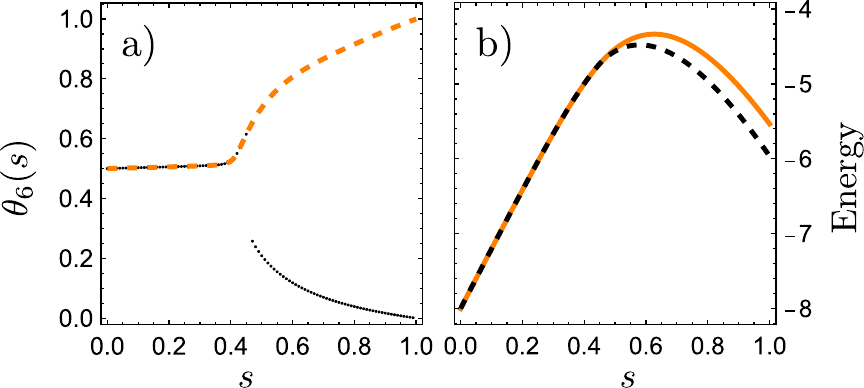}
    \caption{\textit{a}) The orange line shows the adiabatic time evolution of the variational parameter $\theta_6(s)$ on the manifold with $D=1$. The black dots correspond to the instantaneous variational ground state value of the same parameter. We observe a first-order phase transition at $s_c \approx 0.73$. \textit{b}) The orange line corresponds to the energy of the time-dependent variational state. The black line is the energy of the variational ground state. The energy of the time-dependent variational state starts to diverge from the exact variational ground state after the first-order transition.}
    \label{fig: seed 11 theta}
\end{figure}

In contrast, in the example shown in \fref{fig: app sg scaling 2}, even the exact adiabatic evolution ($D=16$) does not converge to the target ground state. In this case, we can explain the failure with the very small minimum gap (shown in \fref{fig: app sg scaling 2}~c)). 

\begin{figure}[!htb]
    \centering
    \includegraphics[width=\columnwidth]{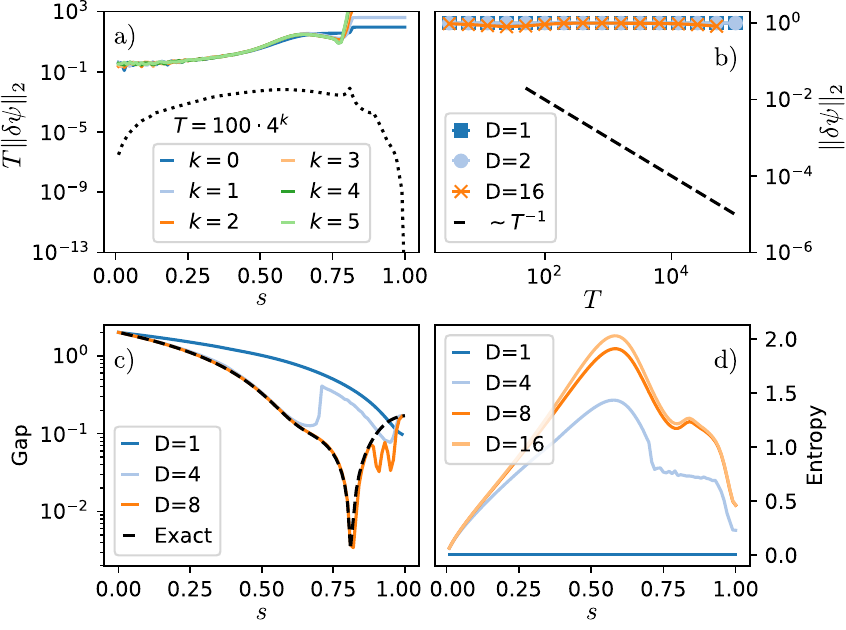}
    \caption{The same as in \fref{fig: app sg scaling 1}, but with a bond dimension $D=8$ for panel \textit{a}. The variational state does not converge to the variational ground state even for a large bond dimension $D=16$, corresponding to the exact solution.
    Due to a small minimum gap, both the adiabatic and variational adiabatic evolution fail to reach the target ground state. 
    }
    \label{fig: app sg scaling 2}
\end{figure}

In \fref{fig: app gap error}, we show the relation between the estimated minimum ground-state gap and the distance $\lVert\delta\psi(1)\rVert_2$ from the target ground state. In most cases, a small minimum ground-state gap is a good indicator of a large final error. Nevertheless, there is a significant number of outliers (shown in red in \fref{fig: app gap error}) for which this simple argument could not detect the failure of the variational approach. 

\begin{figure}[H]
    \centering
    \includegraphics[width=\columnwidth]{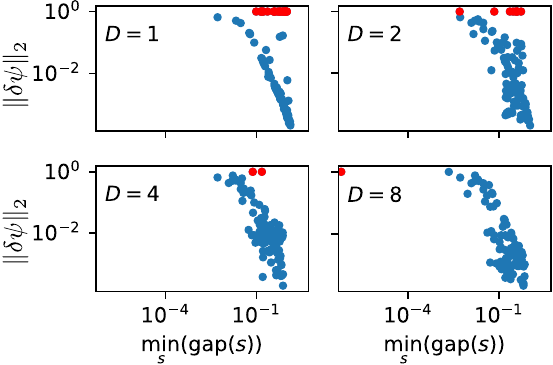}
    \caption{The norm $\lVert\delta\psi(1)\rVert_2$ as a function of the estimated ground-state gap (computed by using the effective local Hamiltonian with different bond dimensions). We show four values of the bond dimension $D=1,2,4,8$.}
    \label{fig: app gap error}
\end{figure}

\end{document}